\documentclass[10pt]{article}
\usepackage{spconf}

\usepackage{verbatim}
\usepackage{amsfonts,amsmath,mathrsfs,amssymb,amsbsy}
\usepackage[final]{graphicx}
\usepackage{times,cite}
\usepackage{caption}
\usepackage{subcaption}

\usepackage{enumitem,kantlipsum}

\long\def\symbolfootnote[#1]#2{\begingroup%
\def\thefootnote{\fnsymbol{footnote}}\footnote[#1]{#2}\endgroup}

\usepackage{verbatim}
\usepackage[]{float,latexsym}
\usepackage{amsfonts,amsmath,mathrsfs,amssymb,amsbsy}
\usepackage{url}
\usepackage{amsthm}

\newtheorem{theorem}{Theorem}[section]

\newtheorem{definition}{Definition}

\usepackage{color}
\definecolor{lightblue}{rgb}{.7, .8, 1}
\definecolor{lightgreen}{rgb}{.6, 1, .6}
\usepackage{color}
\definecolor{brown}{rgb}{1,0.38,0.03}

\definecolor{OliveGreen}{rgb}{.2,0.6,0.2}
\definecolor{BrickRed}{rgb}{.7,0.2,0.2}

\newcommand{\ignore}[1]{} %%% {} empty inside
\long\def\symbolfootnote[#1]#2{\begingroup%
\def\thefootnote{\fnsymbol{footnote}}\footnote[#1]{#2}\endgroup}

\newcommand{\bsp}{\begin{split}}
\newcommand{\esp}{\end{split}}

\begin{document}

\sloppy
\ninept

\title{Cyclostationary Statistical Models and Algorithms for Anomaly Detection Using Multi-Modal Data}

\name{Taposh Banerjee$^{\star}$ \; Gene Whipps$^{\dagger}$ \; Prudhvi Gurram$^{\dagger \ddagger}$ \; and \; Vahid Tarokh$^{\pm}$ 
\thanks{The work of Taposh Banerjee and Vahid Tarokh was supported
by a grant from the Army Research Office, W911NF-
15-1-0479.}
}
\address{$^{\star}$ School of Engineering and Applied Sciences, Harvard University\\
    $^{\dagger}$ U.S. Army Research Laboratory\\
    $^{\ddagger}$ Booz Allen Hamilton\\
    $^{\pm}$Department of ECE, Duke University
    }
%% Create the title:
\maketitle

\begin{abstract}
A framework is proposed to detect anomalies in multi-modal data. A deep neural network-based object detector is employed 
to extract counts of objects and sub-events from the data. A cyclostationary model is proposed to model regular patterns of behavior 
in the count sequences. The anomaly detection problem is formulated as a problem of detecting deviations from learned cyclostationary behavior. 
Sequential algorithms are proposed to detect anomalies using the proposed model. The proposed algorithms are shown to be asymptotically efficient 
in a well-defined sense. The developed algorithms are applied to a multi-modal data consisting of CCTV imagery and social media posts to detect 
a 5K run in New York City. 
\end{abstract}

\begin{keywords}
Nonstationary behavior, change detection, deep neural networks, multi-modal data, count data
\end{keywords}

\section{Introduction}
Many real-life anomaly detection problems including surveillance, infrastructure monitoring, environmental and natural disaster monitoring, border security using unattended ground sensors, crime hot-spot detection for law enforcement, and real-time traffic monitoring involve multi-modal data. For example, in a traffic monitoring application, 
a decision maker who wishes to detect abnormal behavior or impending congestions, may have access to CCTV imagery data, social media data, and other physical sensor data. 
For such applications, efficient algorithms are needed that can detect anomalies or deviations from normal behavior as quickly as possible. 
Effective algorithms can be developed only when one has access to, and has a good understanding of, the multi-modal data encountered in these applications. 
Motivated by this, in this paper, we develop statistical models and algorithms for detecting anomalous behavior in multi-modal data. 
The statistical models studied here are motivated by an analysis of a real-life multi-modal traffic monitoring dataset. 

The datasets studied in this paper were collected by us around a 5K run that occurred in New York City on Sunday, September 24th, 2017. We collected data on two Sundays before the run, and one Sunday after the run. We collected CCTV images and Twitter and Instagram posts over a geographic region from the Red Hook village in Brooklyn on the south end to the Tribeca village on the north end of the collection area. An analysis of the data reveals that the 5K run changes the averages of counts of persons and vehicles appearing in the CCTV cameras and the number of Instagram posts per second posted in the geographical areas near the run. The counts of persons and vehicles appearing in the CCTV images were obtained by passing the images through a convolution neural network-based object detector \cite{bane-fusion-2018}, \cite{fasterrcnn2015}, \cite{vgg162014}, \cite{Everingham10}. See Fig.~\ref{fig:countSeq}. The analysis also suggests that the data has periodic or cyclostationary behavior (see Section~\ref{sec:Data} for more details). In general, in many monitoring applications, a certain cyclostationary behavior is expected, especially while observing long-term patterns of life, unless an unexpected event occurs.
%For example, one would expect the traffic patterns to be similar on Monday mornings, every week, unless an unexpected event occurs. 
\vspace{-0.5cm}
\begin{figure}[H]
\center
\includegraphics[scale=0.2]{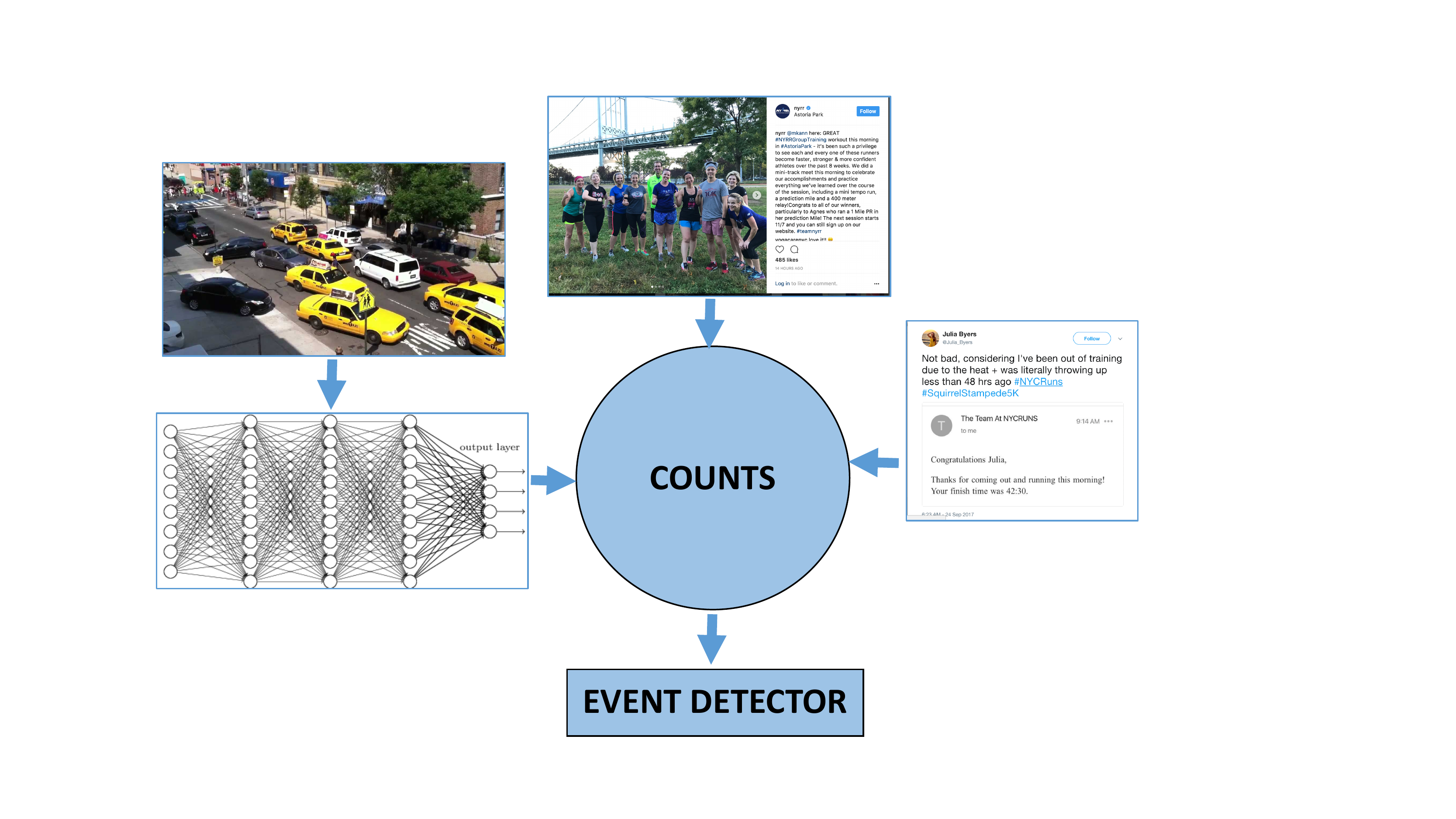}
\caption{Mapping multi-modal data to a sequence of counts}\vspace{-0.2cm}
\label{fig:countSeq}
\end{figure}

In this paper, we define a statistical model to capture the cyclostationary behavior. We also develop sequential algorithms to detect deviations away from learned cyclostationary behavior. 
We develop the sequential algorithms in the framework of quickest change detection \cite{veer-bane-elsevierbook-2013}, \cite{poor-hadj-qcd-book-2009}, \cite{tart-niki-bass-2014}, and also provide their delay and false alarm analysis. 
The salient features of our paper are as follows.
\begin{enumerate}[wide, labelwidth=!, labelindent=0pt]

\item We use a novel framework introduced by us in \cite{bane-fusion-2018} for decision making using multi-modal data involving CCTV images and social media data. In this framework, we use a deep neural network to extract counts of objects from the images. This count is then combined with counts of the number of Tweets and Instagram posts near the CCTV cameras. The decision making is then based on the sequence of counts. 

\item We define the concept of an independent and periodically identically distributed (i.p.i.d) process. We model the count data as an instance of an i.p.i.d. process. We then propose novel algorithms to detect deviations from learned i.p.i.d. behavior. See Definition~\ref{def:ipid}.

\item We define the concept of asymptotic efficiency for a change point detection algorithm and show that our proposed algorithms are asymptotically efficient. See Definition~\ref{def:asyefficient}. 

\item Machine learning and signal processing algorithms for event detection have been developed in the literature \cite{panda2017}, \cite{lee2014}, \cite{szechtman2008}, \cite{neill2007}, \cite{mitchell2013}, \cite{dandrea2015}, \cite{dereszynski2012} \cite{sakaki2010}. However, in these studies, the abnormal event is often either well-defined and/or can be created to train a model. Since the algorithms proposed by us are based on detecting deviations from learned normal behavior, our framework allows for decision making in rare-event scenarios where
the anomalous behavior is hard to learn. 

\end{enumerate}

\vspace{-0.6cm}

\section{Data Analysis}\label{sec:Data}
\vspace{-0.3cm}

Details of the data collected, including information on the deep neural network employed, timings and frame rates can be found in our previous work \cite{bane-fusion-2018}. The objective is to detect the 5K run from the multi-modal data collected. In Figs.~\ref{fig:Person} to Figs.~\ref{fig:Instagram} below, we have plotted averages of the count data collected on the four days, one event day (Sept. 24), and three non-event day (Sept. 10, Sept. 17, and Oct. 1). The data were extracted in 3-second intervals and averaged over a sliding window of size 1000. The figures show plots for two selected cameras: one which was away from the path of the run called the off-path camera, and one which was near the path of the run. The latter is called the on-path camera. 

In Fig.~\ref{fig:PersonOff}, we have plotted the average person count for the off-path camera and in Fig.~\ref{fig:PersonOn}, we have plotted the average person count for the on-path camera. Similar plots for the average vehicle counts are shown in Fig.~\ref{fig:VechicleOff} and Fig.~\ref{fig:VechicleOn}, and for Instagram counts are shown in Fig.~\ref{fig:InstagramOff} and Fig.~\ref{fig:InstagramOn}. The Instagram counts in Fig.~\ref{fig:Instagram} were obtained by averaging the counts for the Instagram posts near the geographical vicinity of the off-path and on-path cameras. 
\begin{figure}
\centering
\begin{subfigure}{.25\textwidth}
  \centering
  \includegraphics[width=\linewidth]{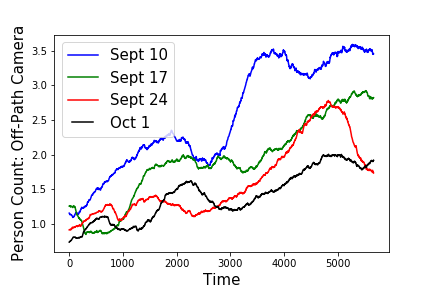}
  \caption{Average person counts for \\ an off-path camera}
  \label{fig:PersonOff}
\end{subfigure}%
%\hspace{1em}
\begin{subfigure}{.25\textwidth}
  \centering
  \includegraphics[width=\linewidth]{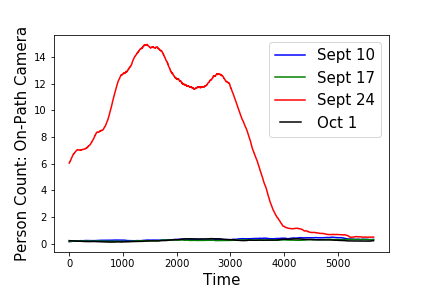}
  \caption{Average person counts for an \\ on-path camera}
  \label{fig:PersonOn}
\end{subfigure}
\caption{Average person counts for the four event days for two cameras: one on the path of the event and one outside the path.}
\label{fig:Person}
\end{figure}
\begin{figure}
\centering
\begin{subfigure}{.25\textwidth}
  \centering
  \includegraphics[width=\linewidth]{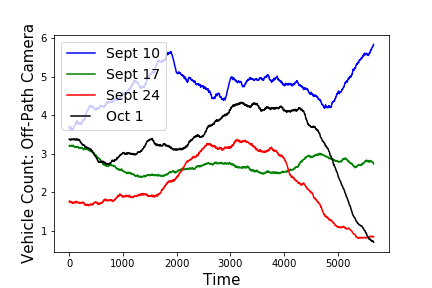}
  \caption{Average vehicle counts for \\ an off-path camera}
  \label{fig:VechicleOff}
\end{subfigure}%
\begin{subfigure}{.25\textwidth}
  \centering
  \includegraphics[width=\linewidth]{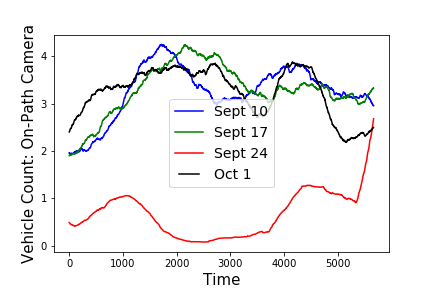}
  \caption{Average vehicle counts for \\ an on-path camera}
  \label{fig:VechicleOn}
\end{subfigure}
\caption{Average vehicle counts for the four event days for two cameras: one on the path of the event and one outside the path.}
\label{fig:Vechicle}
\end{figure}
\begin{figure}
\centering
\begin{subfigure}{.25\textwidth}
  \centering
  \includegraphics[width=\linewidth]{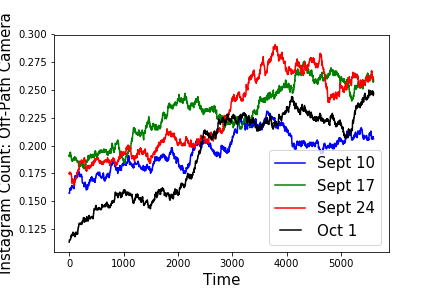}
  \caption{Average Instagram post \\ counts for an off-path camera}
  \label{fig:InstagramOff}
\end{subfigure}%
\begin{subfigure}{.25\textwidth}
  \centering
  \includegraphics[width=\linewidth]{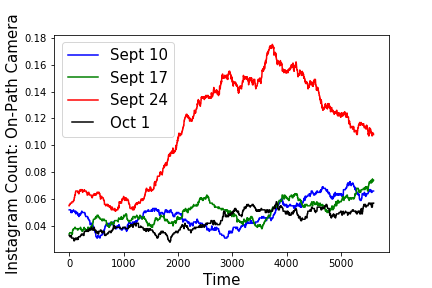}
  \caption{Average Instagram post \\ counts for an on-path camera}
  \label{fig:InstagramOn}
\end{subfigure}
\caption{Average Instagram post counts for the four event days originating near two CCTV cameras: one on the path of the event and one outside the path.}
\label{fig:Instagram}
\end{figure}
We see a clear increase in the average count on the event day for the on-path camera.  Thus, the 5K run event can be detected using the count sequences from both CCTV data and social media posts. More generally, we can expect counts and sequences of sub-events to capture information about anomalous behavior. For example, an event happening twice in a day or two events happening too close to each may indicate a deviation from normal behavior. 

We see from the figures that the data is nonstationary in nature, even 
on non-event days. Also, we observe similarity in statistical behavior in data across all four days from the off-path camera. We also see a similarity in behavior in the data from the on-path cameras on the non-event days. The data also have cyclic behavior. For example, the Instagram count data in Fig.~\ref{fig:InstagramOff} show that the data has a trend 
that repeats itself every Sunday. Thus, the anomaly detection problem here can be rephrased as either the problem of detecting deviations from normal nonstationary behavior 
or as the problem of detecting deviations from normal cyclostationary behavior. 
In \cite{bane-fusion-2018}, we studied a Bayesian problem that captures the problem of detecting changes in the levels of nonstationarity. In this paper, 
we study the latter problem. 

\vspace{-0.4cm}

\section{Mathematical Model and Problem Formulation}\label{sec:Model}
\vspace{-0.2cm}

The central modeling object in this paper is the following.  
\begin{definition}\label{def:ipid}
A stochastic process $\{Y_k\}$ is called independent and periodically identically distributed (i.p.i.d.) if 
the random variables are independent, and there is a positive integer $T$ such that for each $i = 1,  \cdots, T$, 
the process $\{Y_{i + nT}\}_{n=0}^\infty$ is independent and identically distributed (i.i.d). 
\end{definition}
An i.p.i.d. process can be seen as an interleaved version of $T$ i.i.d. stochastic processes, interleaved in a round-robin fashion. An i.p.i.d. process is a wide-sense cyclostationary process \cite{gardner2006cyclostationarity}, but has more structure that we will exploit to develop efficient algorithms. We model a count observation sequence as an i.p.i.d. process. Although counts are discrete in nature, the following discussion is valid for more general random variables as well. 

In our statistical model, the variables $\{Y_k\}$ in the i.p.i.d. process have distribution in a parametric family with parameters $\{\theta_k\}$, and the parameter sequence $\{\theta_k\}$ is periodic with period $T$. In other words, we have a sequence model
\begin{equation}\label{eq:model}\vspace{-0.2cm}
\begin{split}
Y_k & \stackrel{\text{ind}}{\sim} p(\cdot; \theta_k), \; \forall k, \\
\theta_k &= \theta_{k+T},\; \forall k.
\end{split}\vspace{-0.2cm}
\end{equation}
If the data is collected once per hour, then in the above model, the period $T$ would correspond to $T= 24$ hours in a day, and the variables $Y_1, \cdots, Y_{T}$ could correspond to the data collected each hour. In many applications, 
the data is often collected more frequently, at the rate of many samples per second. In such applications, $T$ could be, for example, equal to $24 \times 60 \times 60 \times m$, where $m$ is the number of samples collected per second. Note that the statistical model in \eqref{eq:model} has only $T$ parameters $\theta_1, \cdots, \theta_T$. 

The statistical problem we wish to solve is described as follows. Given the 
parameters $\theta_1, \cdots, \theta_T$, the objective is to observe the process
$\{Y_k\}$ sequentially over time $k$ and detect any changes in the values of any of the parameters. This change has to be detected in real-time with minimum possible delay, subject to a constraint on the rate of false alarms. 
%In this paper, we propose algorithms that can be used to detect such changes. We also provide delay and false alarm analysis for the proposed algorithms. 
The baseline parameters in the problem, the period $T$ and the parameters within a period $\theta_1, \cdots, \theta_T$, can be learned from the training data. General tests for learning an i.p.i.d. process will be reported elsewhere. In this paper, 
we will make additional modeling assumptions to make the learning process simpler. 

Note that the sequence model \eqref{eq:model} studied in this paper is different from the sequence model studied in \cite{Johnstone2015Book} and \cite{tsybakov2009introduction}. In the model studied in \cite{Johnstone2015Book} and \cite{tsybakov2009introduction}, the random variables $\{Y_k-\theta_k\}$ are modeled as Gaussian random variables and the parameters $\{\theta_k\}$ are not periodic. Furthermore, the problem there is of simultaneous estimation of all the different parameters $\{\theta_k\}$ given all the observations $\{Y_k\}$. That is, the problem is not sequential in nature. It is also not a change point problem.

To summarize, in the absence of an anomaly, we model the data as a nonstationary process. But, we believe there is some regularity in the statistical properties of the process. This allows us to model the data as a cyclostationary process. The type of cyclostationary behavior we are interested in is captured by the i.p.i.d. process defined above. The objective in the anomaly detection problem then is to detect a deviation away from a learned cyclostationary or i.p.i.d. behavior.

The algorithm to be used for change detection will depend on the pattern of changes that we assume in the statistical model. We now discuss two change point models for our problem. As discussed above, if the number of samples taken per second is $m$ and the statistical behavior of the data repeats itself after one week, then we have $T=604800m$. 
%In practice, it may be hard to learn, and detect changes in, a large number of parameters. 
In practice, it may be hard to learn a large number of parameters, and detect changes in them.
In order to control the complexity of the problem, we 
assume that the parameters are divided into batches and parameters in each batch are approximately constant. For example, a batch may correspond to data collected in an hour 
and the average count of objects may not change in an hour. Mathematically, we assume that in each cycle or period of length $T$, the vector of parameters
 $\{\theta_k\}_{k=1}^T$ is partitioned into $E$ batches or episodes. 
Specifically, for $N_0=0$ and positive integers $\{N_e\}_{e=1}^E$ we define 
$B_e = \{N_{e-1} + 1, \cdots, N_e\}$ such that 
$\{1, \cdots, T\} = \cup_{e=1}^E B_e, \quad B_e \cap B_f = \emptyset, \text{ for } e \neq f.$
For $e\in \{1, \cdots, E\}$, we define
$\theta_{B_e} = (\theta_{N_{e-1}+1}, \cdots, \theta_{N_e}).$
%$B_1, \cdots, B_E$ 
Thus, $\{\theta_k\}_{k=1}^T$ is partitioned as
\begin{equation}\label{eq:BatchModel}
\overbrace{\theta_1, \cdots, \theta_{N_1}}^{\theta_{B_1}},\overbrace{\theta_{N_1+1}, \cdots, \theta_{N_2}}^{\theta_{B_2}}, \cdots, \overbrace{\theta_{N_{E-1}+1}, \cdots, \theta_{N_E}}^{\theta_{B_E}}. 
\end{equation}
Note that we have $T = \sum_{e=1}^E |B_e|$. 

We further assume a step model for parameters. Under this assumption, the parameters remain constant within a batch resulting in the step-wise constant sequence model
\begin{equation}\label{eq:StepModel}
\overbrace{\theta^{(1)}, \cdots, \theta^{(1)}}^{\theta_{B_1}},\overbrace{\theta^{(2)}, \cdots, \theta^{(2)}}^{\theta_{B_2}}, \cdots, \overbrace{\theta^{(E)}, \cdots, \theta^{(E)}}^{\theta_{B_E}}. 
\end{equation}
That is $\theta^{(1)} = \theta_1=\cdots=\theta_{N_1}$, $\theta^{(2)} = \theta_{N_1+1}=\cdots= \theta_{N_2}$, and so on. Thus, if the batch sizes are large, there are only $E \ll T$ parameters to learn from the data. Also, we have $|B_e|$ samples for batch $e$. 
The objective is then to observe the process $\{Y_k\}$ over time and detect any changes in the parameters $\theta^{(1)}, \cdots, \theta^{(E)}$. 

We now define two change point models. Let $\gamma$ be the change point. If $\gamma=\infty$, i.e., no change occurs, then the stochastic process that we observe, and the parameter values, are given by 
\begin{equation}\label{eq:Normal}%\vspace{-0.3cm}
\begin{split}
\overbrace{\theta^{(1)}, \quad \cdots, \theta^{(1)}}^{\theta_{B_1}}, &\overbrace{\theta^{(2)}, \cdots, \quad \theta^{(2)}}^{\theta_{B_2}} \\ 
{Y_1, \; \; \cdots, Y_{N_1}},& {Y_{N_1+1}, \cdots, Y_{N_2}}, \cdots\\ &\overbrace{\quad \theta^{(E)} \quad, \cdots,   \theta^{(E)}}^{\theta_{B_E}}
\quad \overbrace{\theta^{(1)} \quad , \cdots, \quad \theta^{(1)}}^{\theta_{B_1}}\\ &Y_{N_{T-1}+1}, \cdots, Y_{N_T}, Y_{N_{T}+1}, \cdots, Y_{N_{T+1}}.
\end{split}\vspace{-0.1cm}
\end{equation}
If $\gamma < \infty$, i.e., a change occurs at a finite time $\gamma$, we have two possible change point models. For $k \in \mathbb{N}$, we define the batch of $k$, $b(k)$, as the value $j$ satisfying $(k \text{ mod } T) \in B_j$.
\begin{enumerate}[wide, labelwidth=!, labelindent=0pt]
\item \textit{Change in parameter values in a single batch}: In this model, the distribution of
the random variables $\{Y_k\}$ changes only inside a specific batch say $e \in \{1, \cdots, E\}$. That is, 
in this model, starting at time $\gamma$, the parameter values at all the times change as long as the times fall in the 
batch $e$. Also, the post-change parameter $\lambda_k$ is different for each $k \geq \gamma$, even within a batch. 
Specifically, if $b(k)$ denotes the batch of $k$ then
\begin{equation}\label{eq:CP_single}\vspace{-0.2cm}
\begin{split}
Y_k & \sim p(\cdot\; ; \; \theta^{b(k)}), \text{ for } k < \gamma\\
       & \sim p(\cdot\; ; \; \theta^{b(k)}), \text{ for } k \geq \gamma, \; b(k)\neq e,\\
       & \sim p(\cdot\; ; \; \lambda_k), \text{ for } k \geq \gamma, \; b(k)=e, \text{ with } \lambda_k \neq \theta^{b(k)}. 
\end{split}\vspace{-0.2cm}
\end{equation}
The value of $e$ is not known to the decision maker. 

\item \textit{Change in parameter values in all the batches}:  In this model, the distribution of
the random variables $\{Y_k\}$ changes for all the batches. 
\begin{equation}\label{eq:CP_all}\vspace{-0.2cm}
\begin{split}
Y_k & \sim p(\cdot\; ; \; \theta^{b(k)}), \text{ for } k < \gamma\\
       & \sim p(\cdot\; ; \; \lambda_k), \text{ for } k \geq \gamma, \text{ with } \lambda_k \neq \theta^{b(k)}. 
\end{split}\vspace{-0.2cm}
\end{equation}
\end{enumerate}
In a traffic monitoring scenario, if $T$ corresponds to a day, the single batch change point model may correspond to 
an anomalous behavior between 7 am and 8 am everyday, while the all batch change point may correspond to 
an anomalous behavior throughout the day.

We wish to find a stopping time $\tau$ for the sequence $\{Y_k\}$ so as minimize some version of the average of the detection delay $\tau - \gamma$, 
with a constraint on the false alarm rate. A popular criterion studied in the literature is that by Pollak \cite{poll-astat-1985}
\begin{equation}\label{eq:Pollak}
\begin{split}
\min_\tau &\; \; \sup_\gamma \mathbb{E}_\gamma \; [\tau - \gamma \; | \; \tau > \gamma]\\
\text{Subj. to }& \; \; \mathbb{E}_\infty \; [\tau] \geq \beta,
\end{split}
\end{equation}
where $\mathbb{E}_\gamma$ denotes expectation with respect to the probability measure when the change occurs at time $\gamma$, 
and $\beta$ is a given constraint on the mean time to false alarm. 
Finding optimal solution to such minimax quickest change detection problem is generally hard \cite{veer-bane-elsevierbook-2013}, \cite{poor-hadj-qcd-book-2009}, \cite{tart-niki-bass-2014}. 
We, therefore, propose algorithms (stopping times), and show that they have the following important property, which we also define.  
\begin{definition}\label{def:asyefficient}
A stopping time $\tau$ is called asymptotically efficient for a change point problem, if as $\beta \to \infty$
$$
\mathbb{E}_\infty [\tau] \geq \beta (1 + o(1)),
$$
and there exists a positive constant $C$ such that
$$
\mathbb{E}_1 [\tau] \leq C \log \beta (1 + o(1)). 
$$
\end{definition}
We note that most of the classical optimal algorithms in the literature are asymptotically efficient \cite{veer-bane-elsevierbook-2013}, \cite{poor-hadj-qcd-book-2009}, \cite{tart-niki-bass-2014}, while a trivial algorithm like $\tau \equiv \beta$ is not. Furthermore, according to fundamental limit theorems on change point detection \cite{lai-ieeetit-1998}, 
the performance of any stopping time cannot be of a smaller order of magnitude than $\log \beta (1 + o(1))$. Thus, being asymptotically efficient is an important property to have for a change detection algorithm. 
Comments on optimality with respect to the Pollak's criterion \eqref{eq:Pollak} or Lorden's criterion \cite{lord-amstat-1971} will be provided in an extended version of this paper. 

\vspace{-0.4cm}
\section{Algorithms for Anomaly Detection}\label{sec:Algo}
\vspace{-0.2cm}
The change detection model defined in \eqref{eq:CP_single} and \eqref{eq:CP_all} are similar to change point models studied in sensor network literature
\cite{tart-veer-fusion-2002}, \cite{mei-biometrica-2010}, \cite{tart-veer-sqa-2008}, where a change can affect one, or all the sensors. Observations from a
batch can be viewed as observations from a sensor. The important difference 
between our problem and the sensor network problem is that the decision maker here observes the data from batches in sequence, i.e., does not have access to all
the data at the same time. Nonetheless, the analogy between the two problems provides us with guidelines for identifying relevant algorithms for our problem. 
We will make some assumptions about the way change occurs to simplify our notations, algorithms, and analysis. Algorithms for more general change point models 
can be developed by following the techniques discussed below. 

\vspace{-0.4cm}
\subsection{Algorithm for Detecting Change in a Single Batch}
\vspace{-0.2cm}
We assume that after the change occurs in a single batch $e$, the post-change parameter $\lambda_k$ is the same for all the variables in the batch $e$. 
Since it is not known in which batch $e$ the change occurs, we execute $E$ algorithms, one for each batch, and raise an alarm as soon 
as any of the algorithms detect the change. Mathematically, define the following statistics for data from batch $e$:
\vspace{-0.3cm}
\begin{equation}\label{eq:StatSingle}\vspace{-0.3cm}
\begin{split}
W_n^e = \max_{1 \leq k \leq n} \sup_{\lambda \in \Lambda^{e}} \sum_{i=k: b(i)=e}^n \log [ p(Y_i; \lambda) / p(Y_i; \theta^{(e)}) ],
\end{split}
\end{equation}\vspace{-0.2cm}
where
\begin{equation}\label{eq:Lambda}
\Lambda^{e} = \{\lambda: |\lambda-\theta^{(e)}| \geq \epsilon\}.
\end{equation}
Also, define $\tau^e$ as the stopping time for the batch $e$:\vspace{-0.2cm}
\begin{equation}
\tau^e = \inf\{n \geq 1: W_n^e > A\}.\vspace{-0.1cm}
\end{equation}%\vspace{-0.1cm}
Here, $\epsilon > 0$ is the minimum amount of change from the baseline parameter $\theta^{(e)}$ the algorithm can detect. Note that the condition ${i=k: b(i)=e}$ ensures that only data from the batch $e$ are utilized for computing the statistic $W_n^e$. 
Our change detection algorithm is the minimum of these stopping times. \vspace{-0.2cm}
\begin{equation}\label{stop:one}
\tau_o = \min_{1 \leq e \leq E} \; \tau^e.\vspace{-0.1cm}
\end{equation}
%We have the following theorem. 
\vspace{-0.4cm}
\begin{theorem}\label{thm:single}
Suppose the post-change parameter space $\Lambda^{e}$ in \eqref{eq:Lambda} is finite. Then, the stopping time $\tau_o$ in \eqref{stop:one} is asymptotically efficient. 
\end{theorem}
\vspace{-0.6cm}
\begin{proof}
The false alarm result is true because $\tau^e$ stochastically dominates Lorden's stopping time designed for pre-change parameter $\theta^{(e)}$.
The finite family assumption and martingale arguments imply setting $A = \log \beta (1+o(1))$ will ensure 
$\mathbb{E}_\infty[\tau_o] \geq \beta (1+o(1))$, as $\beta \to \infty$ \cite{tart-veer-fusion-2002} . 
For delay, it can be shown that if $\lambda$ is the true post-change parameter in batch $e$ then as $\beta \to \infty$,
$\mathbb{E}_1[\tau_o]  \leq \frac{\log(\beta) (1 + o(1))}{I(\lambda)} \kappa$,
where $\kappa = (1 + \sum_{f \neq e} |B_f|/|B_e|)$, and $I(\lambda)$ is the Kullback-Leibler divergence between $p(\cdot, \lambda)$ and $p(\cdot, \theta^{(e)})$, implying asymptotic efficiency. 
\end{proof}
\vspace{-0.4cm}
\subsection{Algorithm for Detecting Change in All the Batches}
\vspace{-0.2cm}
We assume that after the change occurs, the post-change parameter $\lambda_k$ is the same for all the variables in a batch $e$. 
Since the change occurs in all the batches, we use an algorithm that combines observations from all the batches. Mathematically, we
compute the statistic
\vspace{-0.2cm}
\begin{equation}\label{eq:StatisticAll}
W_n = \max_{1 \leq k \leq n} \sup_{\lambda^{(e)} \in \Lambda^{(e)}, \; e \leq E} \sum_{i=k}^n \log [ p(Y_i; \lambda^{(b(i))}) / p(Y_i; \theta^{(b(i))}) ],\vspace{-0.2cm}
\end{equation}
and declare an anomaly at the stopping time \vspace{-0.1cm}
\begin{equation}\label{stop:all}
\tau_a = \inf\{n \geq 1: W_n > A\}.\vspace{-0.1cm}
\end{equation}
%We have the following theorem. 
\begin{theorem}\label{thm:all}
Suppose the post-change parameter space $\Lambda^{e}$ in \eqref{eq:Lambda} is finite. Then, the stopping time $\tau_a$ in \eqref{stop:all} is asymptotically efficient. 
\end{theorem}
\begin{proof}
Independence and separation of suprema over $\lambda^{(e)}$ gives $W_n \leq \sum_{e=1}^E W_n^e$. The false alarm result follows from the previous theorem because 
$\{\sum_{e=1}^E W_n^e > A\}$ implies $\{\max_{e=1}^E W_n^e > A/E\}$. 
For the delay analysis, note that removing the maximum operators gives $W_n \geq \sum_{e=1}^E \sum_{i=1: b(i)=e}^n \log [ p(Y_i; \lambda^{(e)}) / p(Y_i; \theta^{(e)}) ]$. 
Asymptotic efficiency follows because the latter's behavior is similar to that of a random walk and based on the arguments in \cite{wood-nonlin-ren-th-book-1982}. 
\end{proof}

\vspace{-0.4cm}
\section{Numerical Results and Conclusions}\label{sec:Numerical}
\vspace{-0.2cm}
We now apply the developed algorithm to the NYC data. Due to a paucity of space, the performance of the algorithm for simulated data will be reported elsewhere. 
We apply $\tau_a$ to the count data because the change appears to affect the entire day's data. In Fig.~\ref{fig:AllTestConcatenatedCount}, 
we have plotted the evolution of the test statistic $W_n$ for all the count data: person count, vehicle count, and the Instagram count. 
In the figure, the data for each modality is arranged in a concatenated fashion, with labeled segments separated via red vertical lines. 
Each day has $6598$ samples. To compute the statistic, we divided the data into four batches, with the first three batches being of length $1500$. We modeled the data as a sequence of Poisson random variables. We used the count data from Sept. 10 (one of the non-event days) to learn the averages of these Poisson random variables for each of the four batches. We assumed that there is only one post-change parameter per batch that is equal to twice the normal parameter (half the normal parameters for vehicles) for that batch. We then applied the test to all the four days of data. In Fig.~\ref{fig:InstTestConcatenatedCount}, we have replotted the test statistic applied to the Instagram counts. As seen from the figures, the algorithm detects the anomaly that occurs on Sept. 24 (event day). 
\begin{figure}
\centering
\hspace{0.5cm}
\begin{subfigure}{.25\textwidth}
  \centering
  \includegraphics[width=\linewidth]{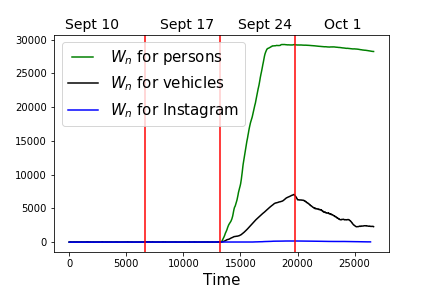}
  \caption{Test statistic $W_n$  for the \\ on-path camera.  }
  \label{fig:AllTestConcatenatedCount}
\end{subfigure}%
%\hspace{1em}
\begin{subfigure}{.25\textwidth}
  \centering
  \includegraphics[width=\linewidth]{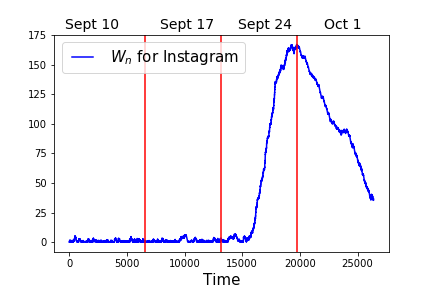}
  \caption{$W_n$ for Instagram \\ replotted.}
  \label{fig:InstTestConcatenatedCount}
\end{subfigure}
\caption{Plots of test statistic $W_n$ from  \eqref{eq:StatisticAll}.}
\label{fig:Numerical}%\vspace{-2cm}
\end{figure}

In future, we will apply the algorithms to other multi-modal datasets to test their effectiveness. We will also study optimality of the proposed algorithms 
for Lorden's and Pollak's criteria. 

%\vspace{-0.2cm}
%\footnotesize
%\pagebreak
%\pagebreak
%\newpage
\clearpage
\newpage
%\nocite{*}
\bibliographystyle{ieeetr}

%\bibliographystyle{elsarticle-harv}

%\bibliographystyle{elsarticle-harv}
%\bibliography{QCD_verVV}
%\newpage

\bibliography{QCD_verSubmitted}

\begin{thebibliography}{10}

\bibitem{bane-fusion-2018}
T.~Banerjee, G.~Whipps, P.~Gurram, and V.~Tarokh, ``Sequential event detection
  using multimodal data in nonstationary environments,'' in {\em Proc. of the
  21st International Conference on Information Fusion}, July 2018.

\bibitem{fasterrcnn2015}
S.~Ren, K.~He, R.~B. Girshick, and J.~Sun, ``Faster {R-CNN:} towards real-time
  object detection with region proposal networks,'' {\em CoRR},
  vol.~abs/1506.01497, 2015.

\bibitem{vgg162014}
K.~Simonyan and A.~Zisserman, ``Very deep convolutional networks for
  large-scale image recognition,'' {\em CoRR}, vol.~abs/1409.1556, 2014.

\bibitem{Everingham10}
M.~Everingham, L.~Van~Gool, C.~K.~I. Williams, J.~Winn, and A.~Zisserman, ``The
  pascal visual object classes (voc) challenge,'' {\em International Journal of
  Computer Vision}, vol.~88, pp.~303--338, June 2010.

\bibitem{veer-bane-elsevierbook-2013}
V.~V. Veeravalli and T.~Banerjee, {\em Quickest Change Detection}.
\newblock Academic Press Library in Signal Processing: Volume 3 -- Array and
  Statistical Signal Processing, 2014.
\newblock \url{http://arxiv.org/abs/1210.5552}.

\bibitem{poor-hadj-qcd-book-2009}
H.~V. Poor and O.~Hadjiliadis, {\em Quickest detection}.
\newblock Cambridge University Press, 2009.

\bibitem{tart-niki-bass-2014}
A.~G. Tartakovsky, I.~V. Nikiforov, and M.~Basseville, {\em Sequential
  Analysis: {Hypothesis} Testing and Change-Point Detection}.
\newblock Statistics, CRC Press, 2014.

\bibitem{panda2017}
R.~Panda and A.~K. Roy-Chowdhury, ``Multi-view surveillance video summarization
  via joint embedding and sparse optimization,'' {\em IEEE Transactions on
  Multimedia}, vol.~19, no.~9, pp.~2010--2021, 2017.

\bibitem{lee2014}
S.~C. Lee and R.~Nevatia, ``Hierarchical abnormal event detection by real time
  and semi-real time multi-tasking video surveillance system,'' {\em Machine
  vision and applications}, vol.~25, no.~1, pp.~133--143, 2014.

\bibitem{szechtman2008}
R.~Szechtman, M.~Kress, K.~Lin, and D.~Cfir, ``Models of sensor operations for
  border surveillance,'' {\em Naval Research Logistics (NRL)}, vol.~55, no.~1,
  pp.~27--41, 2008.

\bibitem{neill2007}
D.~B. Neill and W.~L. Gorr, ``Detecting and preventing emerging epidemics of
  crime,'' {\em Advances in Disease Surveillance}, vol.~4, no.~13, 2007.

\bibitem{mitchell2013}
R.~Mitchell and I.~R. Chen, ``Effect of intrusion detection and response on
  reliability of cyber physical systems,'' {\em IEEE Transactions on
  Reliability}, vol.~62, pp.~199--210, March 2013.

\bibitem{dandrea2015}
E.~D'Andrea, P.~Ducange, B.~Lazzerini, and F.~Marcelloni, ``Real-time detection
  of traffic from {T}witter stream analysis,'' {\em IEEE Transactions on
  Intelligent Transportation Systems}, vol.~16, pp.~2269--2283, Aug 2015.

\bibitem{dereszynski2012}
E.~W. Dereszynski and T.~G. Dietterich, ``Probabilistic models for anomaly
  detection in remote sensor data streams,'' {\em arXiv preprint
  arXiv:1206.5250}, 2012.

\bibitem{sakaki2010}
T.~Sakaki, M.~Okazaki, and Y.~Matsuo, ``Earthquake shakes {T}witter users:
  Real-time event detection by social sensors,'' in {\em Proceedings of the
  19th Int. Conf. on World Wide Web}, pp.~851--860, ACM, 2010.

\bibitem{gardner2006cyclostationarity}
W.~A. Gardner, A.~Napolitano, and L.~Paura, ``Cyclostationarity: Half a century
  of research,'' {\em Signal processing}, vol.~86, no.~4, pp.~639--697, 2006.

\bibitem{Johnstone2015Book}
I.~M. Johnstone, {\em Gaussian estimation: Sequence and wavelet models}.
\newblock Book Draft, 2017.
\newblock Available for download from
  \url{http://statweb.stanford.edu/~imj/GE_08_09_17.pdf}.

\bibitem{tsybakov2009introduction}
A.~B. Tsybakov, {\em Introduction to nonparametric estimation}.
\newblock Springer Series in Statistics. Springer, New York, 2009.

\bibitem{poll-astat-1985}
M.~Pollak, ``Optimal detection of a change in distribution,'' {\em Ann.
  Statist.}, vol.~13, pp.~206--227, Mar. 1985.

\bibitem{lai-ieeetit-1998}
T.~L. Lai, ``Information bounds and quick detection of parameter changes in
  stochastic systems,'' {\em IEEE Trans. Inf. Theory}, vol.~44, pp.~2917
  --2929, Nov. 1998.

\bibitem{lord-amstat-1971}
G.~Lorden, ``Procedures for reacting to a change in distribution,'' {\em Ann.
  Math. Statist.}, vol.~42, pp.~1897--1908, Dec. 1971.

\bibitem{tart-veer-fusion-2002}
A.~G. Tartakovsky and V.~V. Veeravalli, ``An efficient sequential procedure for
  detecting changes in multichannel and distributed systems,'' in {\em {IEEE}
  International Conference on Information Fusion}, vol.~1, (Annapolis, MD),
  pp.~41--48, July 2002.

\bibitem{mei-biometrica-2010}
Y.~Mei, ``Efficient scalable schemes for monitoring a large number of data
  streams,'' {\em Biometrika}, vol.~97, pp.~419--433, Apr. 2010.

\bibitem{tart-veer-sqa-2008}
A.~G. Tartakovsky and V.~V. Veeravalli, ``Asymptotically optimal quickest
  change detection in distributed sensor systems,'' {\em Sequential Analysis},
  vol.~27, pp.~441--475, Oct. 2008.

\bibitem{banerjee2015quickest}
T.~Banerjee, H.~Firouzi, and A.~O. Hero~III, ``Quickest detection for changes
  in maximal knn coherence of random matrices,'' {\em arXiv preprint
  arXiv:1508.04720}, 2015.

\bibitem{wood-nonlin-ren-th-book-1982}
M.~Woodroofe, {\em Nonlinear Renewal Theory in Sequential Analysis}.
\newblock CBMS-NSF regional conference series in applied mathematics, SIAM,
  1982.

\end{thebibliography}

\end{document}